\documentclass[runningheads]{llncs}
\usepackage[T1]{fontenc}
\usepackage{graphicx}
\usepackage{amsmath}
\usepackage{amssymb}

\usepackage{algorithm}
\usepackage{algpseudocode}
\usepackage{todonotes}
\usepackage{xcolor}
\usepackage{comment}
\usepackage{xargs}
\usepackage{caption}
\usepackage{subcaption}
\usepackage{nicefrac}
\usepackage[export]{adjustbox}

\newcommandx{\sal}[2][1=inline]{\todo[linecolor=blue,backgroundcolor=blue!25,bordercolor=blue,#1]{\scriptsize{[SL]~ #2}}}
\newcommandx{\cp}[2][1=inline]{\todo[linecolor=green,backgroundcolor=green!25,bordercolor=green,#1]{\scriptsize{[CP]~ #2}}}

\newcommand{\dimc}[1]{d_{#1}}
\newcommand{\cntSet}[1]{\mathbb{R}^{\dimc{#1}}}
\newcommand{\cntSetDyn}[2]{\mathbb{R}^{\dimc{#1} \times #2}}
\newcommand{\optPair}{(x^*,y^*)}
\newcommand{\optPairTemp}{(\mathbf{x}^*,\mathbf{y}^*)}
\newcommand{\pair}{(x,y)}
\newcommand{\pairTemp}{(\mathbf{x},\mathbf{y})}
\newcommand{\specPair}{(\tilde{x},\tilde{y})}
\newcommand{\specPairTemp}{(\tilde{\mathbf{x}},\tilde{\mathbf{y}})}

\newcommand{\ctt}[2]{u_{#2}^{#1}}
\newcommand{\stt}[2]{z_{#2}^{#1}}

\newcommand{\eqdim}{l}
\newcommand{\ineqdim}{m}

\newcommand{\solset}[1]{\mathsf{SOL}_{#1}}
\newcommand{\solsetc}{\solset{c}}
\newcommand{\solsetu}{\solset{u}}

\newcommand{\calI}{\mathcal{I}}

\newtheorem{assumption}{Assumption}

\begin{document}
\title{When Should Agents Coordinate in Differentiable
Sequential Decision Problems?}
\titlerunning{Coordination in Differentiable
Sequential Decision Problems}

\author{Caleb Probine\orcidID{0009-0002-8395-2666} \and
Su Ann Low\orcidID{0009-0005-7726-8042} \and 
David Fridovich-Keil\orcidID{0000-0002-5866-6441} \and
Ufuk Topcu\orcidID{0000-0003-0819-9985}}
\authorrunning{C. Probine et al.}
\institute{The University of Texas at Austin\\
\email{\{cprobine,suann,utopcu,dfk\}@utexas.edu}}
\maketitle              %
\begin{abstract}
Multi-robot teams must coordinate to operate effectively.
When a team operates in an uncoordinated manner, and agents choose actions that are only individually optimal, the team's outcome can suffer.
However, in many domains,  coordination requires costly communication.
We explore the value of coordination in a broad class of differentiable motion-planning problems. 
In particular, we model coordinated behavior as a spectrum: at one extreme, agents jointly optimize a common team objective, and at the other, agents make unilaterally optimal decisions given their individual decision variables, i.e., they operate at Nash equilibria.
We then demonstrate that reasoning about coordination in differentiable motion-planning problems reduces to reasoning about the \emph{second-order properties} of agents' objectives, and we provide algorithms that use this second-order reasoning to determine at which times a team of agents should coordinate.
\footnote{One can find code to generate all figures and experiments at \url{https://github.com/CalebP547/When-to-Coordinate-in-Differentiable-Sequential-Decision-Problems/}}

\end{abstract}

\section{Introduction}

Coordination is critical in multi-agent systems, from satellite constellations providing telecommunications services to automated vehicles navigating congested roads.
In single-agent optimization and control, coordination is irrelevant, and we typically seek well-defined, optimal solutions. 
However, in a team setting, solutions may exist where each agent makes an optimal decision when the other agents' decisions are fixed---i.e., Nash equilibria---and yet, as shown in Figure~\ref{fig:intro_figure}, the resulting behavior is uncoordinated and can lead to suboptimal outcomes for the team as a whole. 
These uncoordinated Nash equilibrium solutions can be understood to reflect agents' \emph{bounded rationality}: fully rational agents could plausibly identify the best joint decision for the team and play the corresponding strategy, but when agents assess only the impact of changing their own decision variables, the resulting solutions are Nash equilibria.

\definecolor{stargreen}{rgb}{0.1, 0.63, 0.16}
\definecolor{starred}{rgb}{0.9, 0.15, 0.09}
\begin{figure}[t]
     \centering
     \begin{subfigure}{0.38\textwidth}
        \centering
        \includegraphics[width=\textwidth, valign=c]{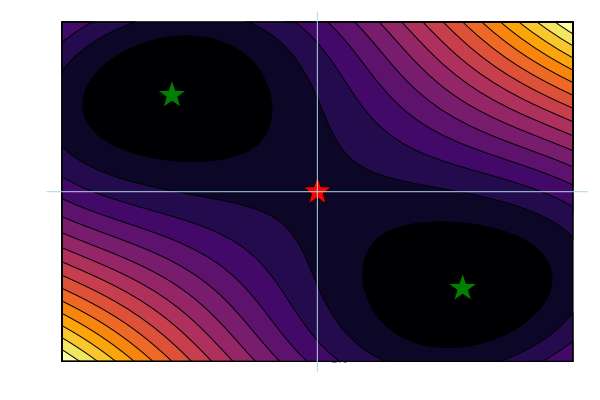}
        \caption{Cost landscape}
        \label{fig:intro_landscape}
    \end{subfigure}
    \hfill
    \begin{subfigure}{0.28\textwidth}
        \centering
        \raisebox{1.4cm}{\includegraphics[width=\textwidth, valign=c]{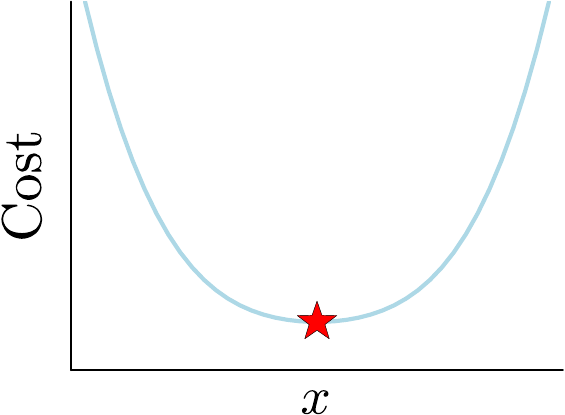}}
        \caption{Cost for $y=0$}
        \label{fig:intro_projection}
    \end{subfigure}
    \hfill
    \begin{subfigure}{0.28\textwidth}
        \centering
        \raisebox{1.5cm}{\includegraphics[width=\textwidth, valign=c]{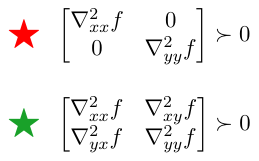}}
        \caption{Hessian structure}
        \label{fig:intro_hessian}
    \end{subfigure}
     \caption{\textbf{Uncoordinated behavior is suboptimal.} 
        Two agents minimize $l(x,y) = \tau\left(x^2 +y^2\right) + \gamma \exp\left(- (\nicefrac{x - y}{\rho}\right)^2 )$.
        (a)~The cost landscape shows two coordinated solutions ({\color{stargreen}$\bigstar$}) at local minima, and one uncoordinated solution ({\color{starred}$\bigstar$}) at a saddle point with higher cost.
        (b)~At the saddle, each agent individually perceives a local minimum: the slice along the $x$-axis (with $y=0$ fixed) shows positive curvature, and vice versa.
        (c)~Solutions are classified by their Hessian structure: coordinated solutions ({\color{stargreen}$\bigstar$}) have a positive definite full Hessian, while uncoordinated solutions ({\color{starred}$\bigstar$}) have only positive definite diagonal blocks.}
     \label{fig:intro_figure}
\end{figure}

We study the problem of deciding to what extent, and at what times, it is valuable for a team of robots to coordinate their decisions in differentiable motion-planning problems.
Coordination requires communication, and communication in multi-agent systems is often costly \cite{carlin2009value}.
Thus, we provide algorithms for unconstrained differentiable motion-planning problems to compute the times at which coordination is most valuable.

To this end, we model coordination as a choice between two coupled optimization problems: one representing coordinated decision-making by the agents, and the other representing uncoordinated decision-making. 
In particular, we follow existing work on the \emph{price of anarchy} \cite{koutsoupias1999worst,bacsar2011prices,ferguson2023collaborative} by associating coordinated behavior with jointly optimal solutions, and uncoordinated behavior with Nash equilibria. 
We then add a coordination action to the model, through which the team can pay to move from uncoordinated to coordinated behavior, and the algorithms we provide determine the optimal level of coordination over time.
To be clear, however, this setting is one of \emph{complete information}: agents have no uncertainty in one another's shared objective, and uncertainty only arises because agents do not know how one another will act.

In smooth optimization problems, the difference between coordinated and uncoordinated behavior amounts to a difference in second-order conditions.
The second-order sufficient condition for a point to be a local minimum of a function is that the Hessian at that point be positive definite \cite[Thm. 2.4]{nocedal2006numerical}.
However, for
a point to be a local Nash equilibrium, only the block diagonal elements of the Hessian corresponding to each agent's decision variables must be positive definite \cite[Thm. 1]{ratliff2016characterization}; the full Hessian can be indefinite, as shown in Figure~\ref{fig:intro_figure}.

We use this second-order reasoning to obtain algorithms for deciding the times at which a team of agents should coordinate in smooth motion-planning problems.
In these dynamic problems, solutions exist on a spectrum where, informally, more coordinated solutions have Hessians with more positive definite structure.
By choosing to coordinate for longer time periods, agents can increase the quality of the solutions they arrive at.

In summary, we make the following contributions.
\paragraph{Contributions}
\begin{itemize}
    \item We formulate coordination in dynamic decision-making problems, i.e., motion-planning, as a choice of the time-steps at which a team should jointly determine their actions.
    We provide a formal characterization of the second-order derivative conditions that describe a spectrum of coordinated solutions, and by coordinating at different time-steps, the agents change the solutions they arrive at, as determined by these second-order conditions.
    \item We formulate a decision problem to balance the cost of coordination and the value it provides, and we provide an algorithm for deciding the time-steps at which a team of agents should coordinate.
    We demonstrate this algorithm in a toy problem, where we observe a diverse range of uncoordinated behaviors, which the agents control by paying for coordination. 
\end{itemize}

\section{Related Work}

We position our approach in the context of three main bodies of literature, namely price of anarchy, equilibrium selection, and bounded rationality in teams.

\paragraph{Price of Anarchy:}

The price of anarchy quantifies the cost of self-interested behavior in multi-agent systems relative to a socially optimal solution \cite{koutsoupias1999worst}. 
While the price of anarchy is a well-studied concept, and existing analyses bound its value in various classes of games \cite{roughgarden2010algorithmic}, most analyses apply to self-interested agents in static settings, while we study dynamic settings with common payoffs.

A handful of works analyze the price of anarchy in common-interest settings.
In particular, existing work makes common-interest assumptions to facilitate analysis of the price of anarchy for strong Nash equilibria \cite{ferguson2023collaborative},
and bounds on the price of anarchy also exist for agents with varying levels of altruism \cite{chen2014altruism}.
However, these works do not apply to dynamic settings. 

Other works study the price of anarchy in dynamic settings, where subtleties of information structure imply finer notions of efficiency.
For example, 
Ba{\c{s}}ar and Zhu \cite{bacsar2011prices} introduce the \emph{price of information} to quantify the difference in equilibrium welfare between feedback and open-loop information patterns. 
Possieri and Sassano \cite{possieri2023measures} also introduce a price of information notion to quantify the differences in welfare between equilibria when agents get observations from different sources in linear quadratic games.
Finally, the \emph{sequential} price of anarchy defines the ratio in social cost between the optimal subgame perfect equilibrium and the social optimum \cite{leme2012curse}.
However, the above works do not focus on common-interest settings.
Furthermore, these works define information structure differences for the whole time horizon, whereas the formulation of dynamic coordination we provide answers the question of \emph{when} it is most valuable for agents to coordinate.
While a further body of work studies the price of anarchy in dynamic settings, cf. \cite{chen2022convergence,zhang2023markov,zyskowski2013price,balandat2013efficiency,parilina2017price,koch2011nash}, these works either do not explore the role of information in the efficiency of equilibria,
or focus on classes of games, such as network flow games, which do not coincide with the differentiable motion-planning problems we study.
Furthermore, the above works again do not focus on the common-interest setting.

While bounds on the price of anarchy exist for urban driving games \cite{zanardi2023bad}, these bounds rely on an open-loop formulation and 
apply only in restricted classes of static congestion games.

More generally, we remark that standard analyses for the price of anarchy simply bound the ratio of costs between two different solution concepts.
In contrast, coordination choices in the formulation we present constitute active decisions to influence whether bad equilibria arise.
In this respect, existing work on optimizing information structures in games is relevant, cf. \cite{hu2024plays}. 
However, that work does not focus on the common-interest setting and does not answer the question of \emph{at what precise times} a group of agents should coordinate.

\paragraph{Equilibrium Selection:} 

In game-theoretic multi-robot interactions, multiple equilibria frequently exist, and work on equilibrium selection provides techniques to ensure robot behavior is robust to these multiple modes of interaction.
Existing methods for game-theoretic motion-planning, 
such as \cite{zhang2025towards,peters2020inference,so2022multimodal,so2023mpogames,bhatt2025multinash,hu2023emergent}, 
handle this uncertainty by inferring the equilibria other agents are playing using techniques such as Bayesian inference \cite{peters2020inference} or opinion dynamics \cite{hu2023emergent}.
We remark that the equilibrium selection problem is also not unique to motion-planning, and also appears in multi-agent reinforcement learning \cite{hu2020other}, where one must ensure that learning does not produce policies that only work for a single equilibrium \cite{hu2020other}.
In contrast to work on equilibrium selection, which develops strategies to ensure that some ego agent can play with other agents following possibly different equilibria, we provide methods to decide on how a team of agents should communicate in order to eliminate high-cost equilibria as possible behavior.

\paragraph{Bounded Rationality in Teams: }

In common-interest games, uncoordinated Nash equilibrium play can be interpreted as boundedly rational behavior for the agents of a team; therefore, we briefly discuss existing work on bounded rationality. 
Common notions of bounded rationality, such as quantal response \cite{mckelvey1995quantal} or cognitive hierarchy models \cite{camerer2004cognitive} express bounded rationality behavior with certain forms of off-equilibrium play.
In contrast, we assume agents will play an equilibrium strategy, but that they may not play an equilibrium strategy that is optimal for the team. 
While existing work on coordination games, for example \cite{bardsley2010explaining}, explores various explanations for why agents with common-interests may select non-optimal joint decisions in common-interest settings, it does not address the differentiable motion-planning setting that we study, nor how agents can optimally communicate to avoid these non-optimal decisions.

\section{Coordination as Second-Order Reasoning}
\label{sec:static_coord_form}

Fundamentally, in continuous settings, the difference between coordinated and uncoordinated solutions corresponds to a difference in second-order properties of the agents' cost functions at the solutions.

Consider a setting where two agents minimize a common objective subject to shared constraints. Specifically, one agent controls variable $x \in \cntSet{x}$ while the other controls $y \in \cntSet{y}$, and both seek to minimize a function $f : \cntSet{x} \times \cntSet{y} \rightarrow \mathbb{R}$ subject to inequality constraints $g(x,y) \geq 0$ and equality constraints $h(x,y) = 0$, where $g : \cntSet{x} \times \cntSet{y} \rightarrow \mathbb{R}^{\ineqdim}$ and $h : \cntSet{x} \times \cntSet{y} \rightarrow \mathbb{R}^{\eqdim}$. 
For index $i\in \{1,\ldots,\ineqdim\}$, $g_i$ is the $i^{\text{th}}$ inequality constraint, i.e., the function $g_i: \cntSet{x} \times \cntSet{y} \rightarrow \mathbb{R}$ that is equal to the $i^{\text{th}}$ element of $g$.
Similarly, for $j \in \{1,\ldots,\eqdim\}$, $h_j$ is the $j^{\text{th}}$ equality constraint.
We make the following blanket assumption regarding the problem.
\begin{assumption}
    Functions $f$,$g$, and $h$ have continuous second derivatives throughout their entire domains.
\end{assumption}

We follow existing work on the \emph{price of anarchy} \cite{koutsoupias1999worst,bacsar2011prices,ferguson2023collaborative} by associating coordinated behavior with optima and uncoordinated behavior with Nash equilibria.
In particular, a coordinated solution is a local solution to the optimization problem in which both agents' variables are decision variables.
\begin{definition}[Coordinated solution]
    A point $\optPair$ is a coordinated solution if it is a locally optimal solution to the joint optimization problem
    \begin{equation}
        \min_{x,y} \ f(x,y) \ \textrm{ \emph{subject to} } \ g(x,y) \geq 0, \ h(x,y) = 0.
    \end{equation}
    That is, there exists some $\epsilon > 0$ such that 
    \begin{equation}
        f \optPair \leq f(x,y) \quad \forall (x,y) : g(x,y) \geq 0, h(x,y) = 0, ||(x,y) - (x^*,y^*)|| \leq \epsilon.
    \end{equation}
\end{definition}

Before we formally define uncoordinated solutions, we provide background on first-order optimality conditions for nonlinear optimization.
Classical optimization theory provides first-order necessary conditions for a point to be a local optimum in a constrained optimization problem.

\begin{theorem}[First-order necessary conditions {\cite[Thm. 12.1]{nocedal2006numerical}}]
    Suppose $z^*$ is a local solution to 
    \begin{equation}
        \min_{z} \ f(z) \ \textrm{ \emph{subject to} } \ g(z) \geq 0, \ h(z) = 0,
    \end{equation}
    and additionally assume that an appropriate constraint qualification, such as the linear independence constraint qualification \cite[Def. 12.4]{nocedal2006numerical} holds at $z^*$.
    Then there exists multipliers $\lambda^* \in \mathbb{R}^{\ineqdim}$ and $\mu^*\in\mathbb{R}^{\eqdim}$ such that
    \begin{align}
        & \nabla_z \left( f + (\lambda^*)^\top g + (\mu^*)^\top h \right)\vert_{z^*} = 0, \\
        & 0 \leq g(z^*) \perp \lambda^* \geq 0, \\
        & h(z^*) = 0.
    \end{align}
\end{theorem}

Uncoordinated solutions correspond to Nash equilibria where agents share common multipliers.
\begin{definition}[Uncoordinated solution]
    \label{def:uncoord_static}
    A point $\optPair$ is an uncoordinated solution if 
    \begin{enumerate}
        \item The first agent's variable $x^*$ is a local solution to
        \begin{equation}
            \min_{x} \ f(x,y^*) \ \textrm{ \emph{subject to} } \ g(x,y^*) \geq 0, \ h(x,y^*) = 0,
        \end{equation}
        and the second agent's variable $y^*$ is a local solution to
        \begin{equation}
            \min_{y} \ f(x^*,y) \ \textrm{ \emph{subject to} } \ g(x^*,y) \geq 0, \ h(x^*,y) = 0.
        \end{equation}
        \item Common Lagrange multipliers exist for the agents' problems. 
        That is, there exist Lagrange multipliers $\lambda^* \in \mathbb{R}^{\ineqdim}$ and $\mu^* \in\mathbb{R}^{\eqdim}$ such that 
        \begin{align}
        & \nabla_{\pair} \left( f + (\lambda^*)^\top g + (\mu^*)^\top h \right)\vert_{\optPair} = 0, \\
        & 0 \leq g\optPair \perp \lambda^* \geq 0, \label{eq:fo_unc_compl} \\
        & h\optPair = 0.
    \end{align}
    \end{enumerate}
\end{definition}
\emph{Strictly uncoordinated solutions} are those points that are uncoordinated solutions, but not coordinated solutions. 

The requirement of common Lagrange multipliers is common in the literature.
Specifically, this condition implies that uncoordinated solutions are variational equilibria, a subset of generalized Nash equilibria characterized by precisely this restriction to a shared set of multipliers \cite{facchinei2007generalized,kulkarni2009new}.

We remark that while all coordinated solutions satisfying appropriate constraint qualifications are uncoordinated solutions, there \emph{are} uncoordinated solutions which are not coordinated.
More formally, let $\solsetc$ and $\solsetu$ denote the sets of coordinated and uncoordinated solutions, respectively.
At any solution in $\solsetc$, no agent can locally reduce the cost, and furthermore, if the linear independence constraint qualification holds at the solution, Lagrange multipliers will exist \cite[Thm. 12.1]{nocedal2006numerical}.
Thus, any point in $\solsetc$ satisfying appropriate constraint qualifications lies in $\solsetu$.
However, the converse does not hold as we demonstrated in Figure~\ref{fig:intro_figure}, where we see that uncoordinated solutions may exist that have higher costs than coordinated solutions.

\subsection{The Shape of the Set of Solutions}

To formulate the coordination problem, we must first characterize the geometry of the set of uncoordinated solutions.
It is initially unclear whether the set of uncoordinated solutions is a connected set we should traverse using standard nonlinear optimization techniques, or a set of isolated points about which we should reason discretely.

Lemma~\ref{lem:uncoord_connec} indicates that coordination is a discrete problem.
\begin{lemma}
    \label{lem:uncoord_connec}
    Let $c: [0,1] \rightarrow \cntSet{x} \times \cntSet{y}$ be a continuously differentiable path comprising uncoordinated solutions, i.e. $c(t)$ satisfies Definition~\ref{def:uncoord_static} for all $t \in [0, 1]$. Then $f(c(t_1)) = f(c(t_2)), $ for all $ t_1, t_2 \in [0, 1]$.
\end{lemma}
In the unconstrained case, this result follows directly from the chain rule and the fundamental theorem of calculus.
In particular, the total derivative $\nicefrac{df}{dt}$ is
\begin{equation}
    \frac{df}{dt} = \left\langle \frac{dx}{dt}, \nabla_x f(x(t), y(t)) \right\rangle + \left\langle \frac{dy}{dt}, \nabla_y f(x(t), y(t)) \right\rangle = 0
\end{equation}
for $c(t) = (x(t),y(t))$, 
where $\nabla f$ is zero as a consequence of first-order optimality.
For the constrained case, we only have that the gradient of the Lagrangian $\mathcal{L} = f + \lambda^\top g + \mu^\top h$ is zero. 
However, we can use the feasibility of the points on the path to argue that $\langle \dot{c}, \nabla h_j \rangle$ and $\lambda_i \langle \dot{c}, \nabla g_i \rangle$ are zero at each point on the path, which in turn implies $\langle \dot{c}, \nabla f \rangle = 0$.
We give a full proof in the appendix. 
Lemma~\ref{lem:uncoord_connec} in turn implies that, for a set of points $S$ connected by smooth paths, the objective must be constant on this set.
Thus, if we optimize over the set of points satisfying first-order conditions in a way that stays in such a connected set, we will only find points with the same objective values.

\subsection{Classifying Behavior with Second-order Conditions}
\label{sec:static_coord_solve}

We now discuss how, given a discrete set of uncoordinated solutions, we can use second-order properties to identify the coordinated solutions in this set.
For clarity, we focus on unconstrained problems and briefly discuss the generalization to the constrained case.

Coordinated and uncoordinated solutions satisfy identical first-order optimality conditions.
In particular, in the unconstrained setting, both coordinated and uncoordinated solutions satisfy $\nabla f(x,y) = 0$.

Given a candidate solution $\specPair$ satisfying the first-order conditions, we classify it as coordinated or uncoordinated by examining the structure of the Hessian of $f$ at $\specPair$. In particular, the second-order sufficient conditions for minima in unconstrained optimization \cite[Thm. 2.4]{nocedal2006numerical} imply the following second-order sufficient conditions for coordinated and uncoordinated solutions.
\begin{equation}
     \underbrace{\begin{bmatrix}
        \nabla^2_{xx} f & \nabla^2_{xy} f \\
        \nabla^2_{yx} f & \nabla^2_{yy} f \\
    \end{bmatrix} \succ 0}_{\text{coordinated solution}}
    \hspace{40pt}
     \underbrace{\begin{bmatrix}
        \nabla^2_{xx} f  & 0 \\
        0 & \nabla^2_{yy} f \\
    \end{bmatrix} \succ 0.}_{\text{uncoordinated solution}}
\end{equation}
Importantly, a coordinated solution requires the \emph{full} Hessian to be positive definite, while uncoordinated solutions only require that the \emph{diagonal blocks} for each agent's variables be positive definite. 
The off-diagonal cross-terms may introduce indefiniteness in the full Hessian, as seen in Figure~\ref{fig:intro_figure}, where the strictly uncoordinated solution is a saddle point. 
We use these second-order conditions to classify each first-order solution $(\tilde x, \tilde y)$ we find. 
If a point does not satisfy either condition, i.e., one of the diagonal blocks is not positive definite, we discard it.

\begin{remark}
This second-order classification method relies on sufficient conditions for local optimality, and as a result, we discard solutions that satisfy positive \emph{semi}-definiteness.
However, as computing whether a point is a local minimum is an NP-complete problem \cite{murty1985some}, checking positive definiteness represents an efficient proxy for checking local optimality.
\end{remark}

\paragraph{Extensions to constrained problems.}
We can extend the algorithm to constrained settings with appropriate changes to first and second-order reasoning.
In particular, we replace the root-finding algorithm with a mixed complementarity problem \cite{facchinei2003finite} solver such as PATH \cite{dirkse1995path} to handle the nonsmoothness introduced by the complementarity conditions in \eqref{eq:fo_unc_compl}.
For the second-order conditions, we replace Hessian checks with a check that the Lagrangian's Hessian is positive definite in the directions defined by the critical cone~\cite[Thm. 12.6]{nocedal2006numerical}.

\subsection{Illustrative Example: Robot Separation}

We demonstrate the differences between coordinated and uncoordinated solutions on a problem involving two robots that seek to maximize their separation while incurring control costs. 
Agent one chooses a scalar $x \in \mathbb{R}$, while agent two chooses a scalar $y \in \mathbb{R}$, and the joint cost function is

\begin{equation}
    l(x,y) = \tau\left(x^2 +y^2\right) + \gamma \exp\left(- \left(\frac{x - y}{\rho}\right)^2 \right).
\end{equation}
The first term penalizes control effort with weight $\tau > 0$, while the second term penalizes proximity with the exponential term decreasing to zero as the robots move apart.
The constant $\rho> 0$ controls the length scale of the separation term.

This cost function exhibits the following coordinated and uncoordinated solutions, which we show in Figure~\ref{fig:intro_figure} for the case of $\tau = 0.5$, $\gamma = 1.0$, and $\rho = 1.5$.
\begin{itemize}
    \item \textbf{Two coordinated solutions} at $(x, y) \approx (\pm c^*, \mp c^*)$ for some $c^* > 0$, each achieving objective value $f = 0.886$. These correspond to the robots jointly agreeing to move in opposite directions, one left and one right, achieving separation while sharing the control burden.
    \item \textbf{One uncoordinated solution} at $(x, y) = (0, 0)$ with objective value $f = 1.0$. At this Nash equilibrium, neither robot wishes to unilaterally expend control effort. Moving alone increases control cost without sufficiently reducing the proximity penalty, since the other robot remains stationary.
\end{itemize}

\section{Formulating and Solving Dynamic Coordination Problems}

We formulate the problem of deciding \emph{when} to coordinate by generalizing the above second-order reasoning.
In particular, when agents make decisions over some time horizon, solutions can exhibit coordination on some time intervals but not others, and thus, we can make finer distinctions than the binary classification of coordinated vs. uncoordinated that we define in Section~\ref{sec:static_coord_form}.

As in Section~\ref{sec:static_coord_form}, we assume that two agents minimize a common cost function, and we additionally assume that the agents make their decisions across a series of time-steps.
In particular, agent one chooses a sequence of decisions $\mathbf{x} = (x_1,\ldots,x_T)$ for $x_t \in \cntSet{x}$, and similarly, agent two chooses a sequence $\mathbf{y}$, where $y_t \in \cntSet{y}$.
The agents choose $\mathbf{x}$ and $\mathbf{y}$ to minimize $f : \cntSetDyn{x}{T} \times \cntSetDyn{y}{T} \rightarrow \mathbb{R}$.
In this setting, we observe that the Hessian of $f$ has the following block structure.
\begin{equation}
    \nabla^2 f = \begin{bmatrix}
        \nabla^2_{x_1 x_1} f & \nabla^2_{x_1 x_2} f & \cdots & \nabla^2_{x_1 y_T} f \\
        \nabla^2_{x_2 x_1} f & \nabla^2_{x_2 x_2} f & \cdots & \nabla^2_{x_2 y_T} f \\
        \vdots & \vdots & \ddots & \vdots \\
        \nabla^2_{y_T x_1} f & \nabla^2_{y_T x_2} f & \cdots & \nabla^2_{y_T y_T} f
    \end{bmatrix}.
\end{equation}

We formalize coordination in dynamic settings through positive-definiteness requirements on subblocks of the Hessian corresponding to specific time intervals.
\begin{definition}[Dynamic coordination intervals]
    Let $S \subseteq \{1,\ldots, T\}$ be a set of time indices, and let $\optPairTemp$ be a point satisfying the first-order optimality conditions. We say that the agents are \emph{coordinated on the set $S$} if:
    \begin{equation}
        \label{eq:set_Hessian}
        \begin{bmatrix}
            \nabla^2_{\mathbf{x}_S \mathbf{x}_S} f \optPairTemp & \nabla^2_{\mathbf{x}_S \mathbf{y}_S} f \optPairTemp \\
            \nabla^2_{\mathbf{y}_S \mathbf{x}_S} f \optPairTemp & \nabla^2_{\mathbf{y}_S \mathbf{y}_S} f \optPairTemp
        \end{bmatrix} \succ 0,
    \end{equation}
    \begin{equation}
        \label{eq:ind_hessian}
        \nabla^2_{\mathbf{x}\mathbf{x}} f \optPairTemp \succ 0, \quad \text{and} \quad 
        \nabla^2_{\mathbf{y}\mathbf{y}} f \optPairTemp \succ 0,
    \end{equation}
    where $\mathbf{x}_S$ and $\mathbf{y}_S$ denote the decision variables restricted to time-steps in $S$.
\end{definition}
Under this definition, agents that are coordinated on the set $S$ have solutions that are jointly optimal when restricted to the decision variables for times in $S$, as encoded in \eqref{eq:set_Hessian}, and this fact follows from second-order sufficient conditions for unconstrained optimization \cite[Thm. 2.4]{nocedal2006numerical}.
The requirement that the diagonal blocks $\nabla^2_{\mathbf{xx}} f$ and $\nabla^2_{\mathbf{yy}} f$ remain positive definite, i.e. \eqref{eq:ind_hessian}, ensures the solution remains unilaterally optimal for each individual agent. 
We denote by $\solset{S}$ the set of solutions coordinated on the interval $S$. 
The set $\solset{\emptyset}$ contains all first-order solutions with positive definite subblocks for the individual agents, i.e., points satisfying \eqref{eq:ind_hessian} and $\nabla f \pairTemp = 0$.

\begin{remark}
    We emphasize that solutions that are coordinated on a set of subsets $\{S_1, S_2, \dots, S_k\}$ need not be coordinated on the union $\bigcup_{j} S_j$ of these subsets, and we provide a small example for intuition.
    Consider an example where $$f\pairTemp = [\mathbf{x}^\top \mathbf{y}^\top]^\top Q \begin{bmatrix}
        \mathbf{x}\\
        \mathbf{y}
    \end{bmatrix},~\text{with}~Q = \frac{1}{5}\left[\begin{array}{@{}ccc|ccc@{}}
    5 & & & -2 & -2 & -2\\
     & 5 & & -2 & -2 & -2\\
     &  & 5 & -2 & -2 & -2\\\hline
    -2 & -2 & -2 & 5 &  & \\
    -2 & -2 & -2 & & 5 & \\
    -2 & -2 & -2 & & & 5
  \end{array}\right].$$
    The solution $(\mathbf{x}^*,\mathbf{y}^*) = \mathbf{0}$ lies in $\mathsf{SOL}_S$ for all $S$ of cardinality at most $2$, but not in $\mathsf{SOL}_{\{1,2,3\}}$.
    That is, agents are coordinated on every time pair, but not on the set of all times.
\end{remark}

\subsection{Formulating Dynamic Coordination Costs}

In this dynamic motion-planning setting,
the team chooses a set of times at which to coordinate.
In particular, we let $ \calI$ represent the set of intervals in $\{1,\ldots,T\}$, and we represent the team's decision of when to coordinate as a probability distribution on $\calI$.
One could also take $\cal I$ to be the power set $2^{\{1,\ldots,T\}}$, with the caveat that the exponential size of this set will be a computational burden.

In the problem of deciding when to coordinate, the agents' overall objective comprises a term for coordination costs and a term to represent the cost of the solutions that result when the agents apply a given level of coordination.

When the team coordinates on the interval $S \in \mathcal{I}$, they pay a cost corresponding to the average cost of all solutions that are coordinated on that interval.
In particular, the team will pay a cost of
\begin{equation}
        \overline{f}(S) = \frac{1}{|\solset{S}|} \sum_{\specPairTemp \in \solset{S}} f \specPairTemp
\end{equation}
due to the solutions they arrive at.
\begin{remark}
    This formulation implicitly assumes a uniform distribution on the solutions coordinated on $S$, but in principle, one can replace this uniform distribution with any  other plausible model of why agents might choose among particular solutions that are coordinated on the interval $S$.
\end{remark}

To define the cost of coordination itself, we first model the team as having a nominal distribution over the intervals in $\calI$ over which they are coordinated.
In particular, let $\Delta^\calI$ be the set of distributions on $\calI$.
We assume that, for some $q \in \Delta^\calI$, each interval $S\in \calI$ occurs with probability $q_S$ if the team does nothing, in a manner that models the team's natural communication availability. 
For example, if a team of ground robots starts their mission from a common base station, communication availability might be higher at the start of the mission when the agents are close to each other.
We assume a uniform distribution on $\calI$, but note that the ensuing framework will generalize readily to other distributions.

The coordination problem involves choosing a distribution $p$ on intervals to minimize the cost of the resulting solutions while staying close to the aforementioned nominal distribution.
Formally, the team's coordination problem is 
\begin{equation}
    \label{eq:dyn_stack_coordination}
    \min_{p \in \Delta^\calI} \sum_{S \in \calI} \left(c_S (p_S - q_S)^2 + p_S \overline{f}(S)\right).
\end{equation}
The $c_S (p_S - q_S)^2$ term is the coordination cost, where the value $c_S$
encodes the cost of changing the probability of coordination on the interval $S$, and we set this term such that $c_S = |S|$ for each interval $S$.
This choice of $c_S$ models the notion that keeping a communication link open for a longer period of time leads to higher communication costs.

\subsection{Solving Dynamic Coordination Problems}

We construct and solve dynamic coordination problems using the high-level procedure that we detail below.
\begin{enumerate}
    \item Generate candidate solutions that satisfy first-order optimality conditions. 
    \item Classify these solutions by the intervals on which the agents are coordinated.
    \item Solve the dynamic coordination problem \eqref{eq:dyn_stack_coordination}.
\end{enumerate}
In particular, any candidate solution for the agents (i.e., an element of $\solset{\emptyset})$, satisfies the first-order condition $\nabla f \pairTemp = 0$.
We can in turn identify roots of $\nabla f \pairTemp$ with classical Newton algorithms, cf. \cite[Ch. 11]{nocedal2006numerical}, and we can easily generate a set of diverse solutions by initiating the Newton solver from a random set of initializations.
This approach of generating solutions using random initializations follows existing work on equilibrium selection~\cite{peters2020inference}.
Although this solution-generation approach does not guarantee the discovery of all solutions, increasing the number of random initializations improves coverage of the solution set.
For each generated first-order solution, we can then easily classify it by the intervals on which the agents are coordinated using the second-order conditions in \eqref{eq:set_Hessian} and \eqref{eq:ind_hessian}, following which we solve the coordination problem \eqref{eq:dyn_stack_coordination}.
Algorithm~\ref{alg:dynamic_coordination} outlines pseudocode for this procedure.

\begin{algorithm}[t]
\caption{Evaluating the value of coordination in dynamic settings}
\label{alg:dynamic_coordination}
\begin{algorithmic}[1]
\State Initialize $\solset{S}$ to $\emptyset$ for all intervals $S \in \calI$
\For{$i = 1,\ldots, N$}
    \State Solve $\nabla f\pairTemp = 0$ using randomly initialized Newton method to get $\specPairTemp$
    \State Check that $\specPairTemp$ is not a duplicate of a previous solution
        \If{$\nabla^2_{xx} f \specPairTemp \succ 0$ and $\nabla^2_{yy} f \specPairTemp \succ 0$}
            \ForAll{intervals $S \in \calI$}
                \State $H_S \gets $ subblock of the Hessian corresponding to $S$
                \If{$H_S \succ 0$}
                    \State $\solset{S} \gets \solset{S} \cup \{ \specPairTemp \}$ 
                \EndIf
            \EndFor
        \EndIf
\EndFor
\State Solve \eqref{eq:dyn_stack_coordination} and return optimal solution $p^*$
\end{algorithmic}
\end{algorithm}

\section{Empirical Results}

We demonstrate the proposed dynamic coordination formulation and algorithm through a generalization of the example in Section~\ref{sec:static_coord_solve}.
In particular, we consider a problem where two robots move in one dimension following discrete-time single-integrator dynamics, and the robots both wish to minimize cumulative control costs and proximity costs, as below.

\begin{align}
    \min_{\mathbf{u}^1 \in \mathbb{R}^T, \mathbf{u}^2 \in \mathbb{R}^T} \ \ &  \sum_{t=1}^T \tau\big((\ctt{1}{t})^2 +(\ctt{2}{t})^2\big) + \sum_{t=1}^{T+1} \gamma \exp\left(- \left(\frac{\stt{1}{t} - \stt{2}{t}}{\rho}\right)^2 \right) \\
    \text{s.t.} \ \ & \stt{i}{1} = 0,~\forall i \in \{1, 2\},\\
                    & \stt{i}{t+1} = \stt{i}{t} + \ctt{i}{t}, \quad \forall t \in \{1,\ldots, T\},~\forall i \in \{1, 2\}.
\end{align}
In this scenario, agent one controls a sequence of decisions $\mathbf{u}^1 = (u^{1}_1,\ldots,u^{2}_T)$, and agent two controls $\mathbf{u}^2$, a similarly structured sequence of decisions.
We rewrite this problem in an unconstrained form by substituting the dynamics constraints into the objective so that we can apply Algorithm~\ref{alg:dynamic_coordination}.
We study this problem for two cases with $T$ equal to $6$ and $10$, respectively.

In the setting with $ T=6$, we observe that coordinated solutions can be low-cost, while uncoordinated solutions are suboptimal and involve agents swapping positions unnecessarily.
We show these solutions in Figure~\ref{fig:example_T_6}.

\begin{figure}  
     \centering
     \begin{subfigure}[b]{0.32\textwidth}
         \centering
         \includegraphics[width=\textwidth]{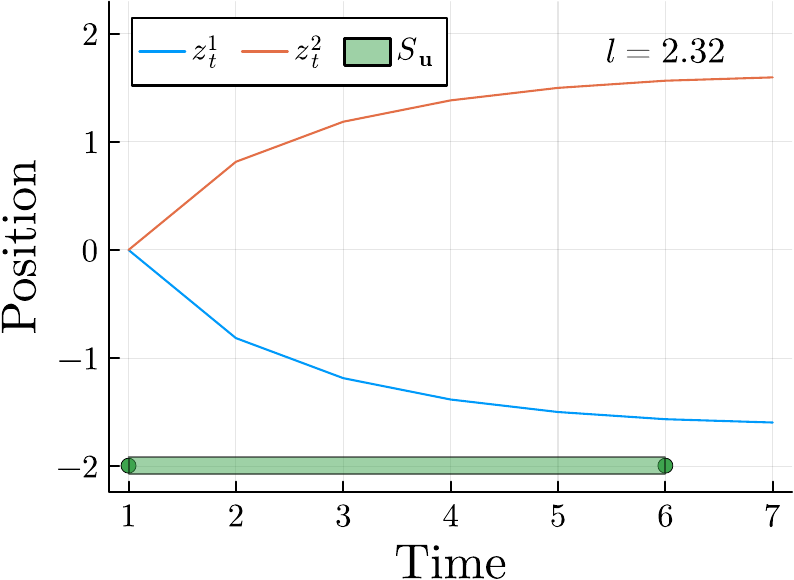}
         \caption{}
         \label{fig:coord_1_T_6}
     \end{subfigure}
     \hfill
     \begin{subfigure}[b]{0.32\textwidth}
         \centering
         \includegraphics[width=\textwidth]{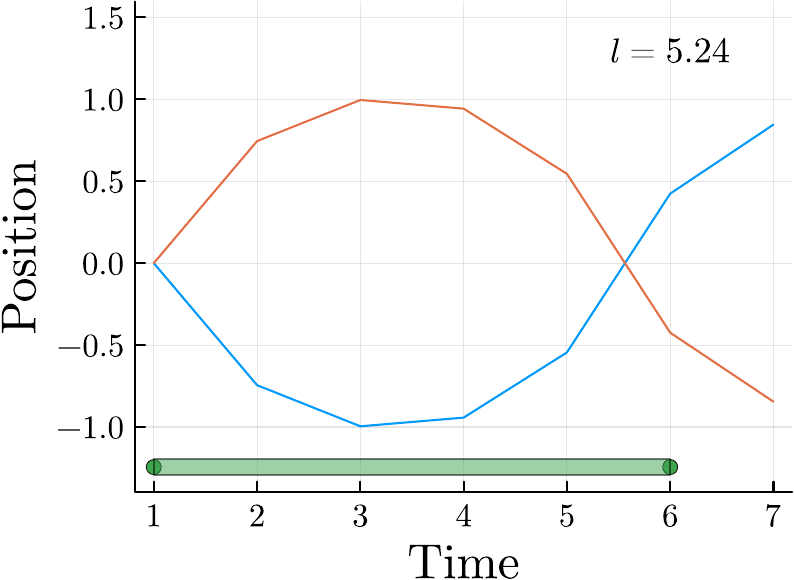}
         \caption{}
         \label{fig:coord_2_T_6}
     \end{subfigure}
     \begin{subfigure}[b]{0.32\textwidth}
         \centering
         \includegraphics[width=\textwidth]{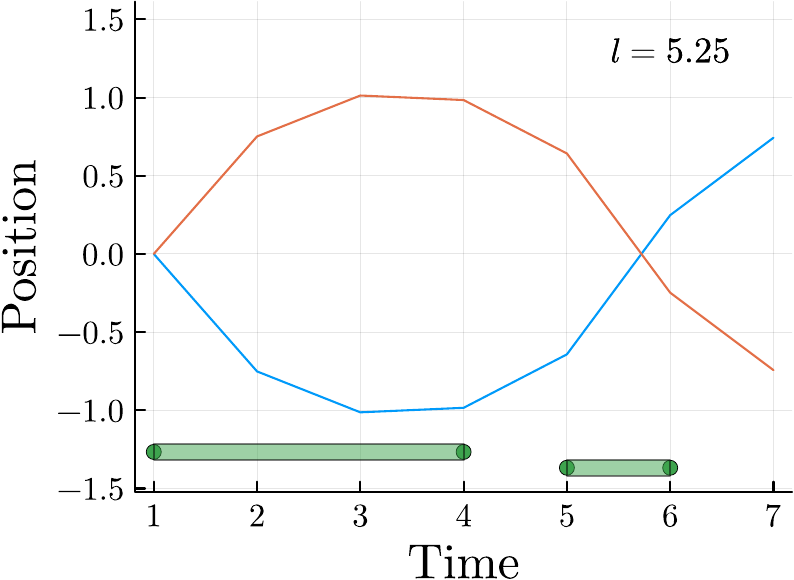}
         \caption{}
         \label{fig:uncoord_1_T_6}
     \end{subfigure}
        \caption{
        \textbf{Coordinated behavior is valuable in dynamic settings.}
        We plot $z_t$ for the three first-order solutions that we find by randomly initializing a root-finding algorithm. 
        Agents one and two follow the blue and orange trajectories, respectively, and the green bars represent the set of intervals $S_\mathbf{u} = \{S ~|~ \mathbf{u} \in \mathsf{SOL}_S\}$, i.e., the intervals during which the agents are coordinated for each solution.
        We remove dominated sets of time intervals when plotting $S_\mathbf{u}$, for example, we only plot the interval $\{1,\ldots,6\}$ for the solution in (a), even though the solution is coordinated on any subset $S$.
        We also remark that we also found solutions where the trajectories of agents one and two were swapped, which we do not depict.
        We additionally annotate the cost $l$ of each trajectory.
        There are two coordinated solutions of differing costs shown in (a) and (b).
        There is an additional uncoordinated solution that the agents may reach if they do not coordinate on the correct time-steps.
        Thus, by coordinating, the team can reduce the chance of reaching a high-cost solution.
        Note that the solutions in (b) and (c) are different solutions, and this difference is most apparent at $T = 6$. 
        }
        \label{fig:example_T_6}
\end{figure}

The optimal solution to the corresponding dynamic coordination problem for $T = 6$ places the most probability on the shortest interval that differentiates coordinated and uncoordinated solutions.
Figure~\ref{fig:heatmap_T_6} shows that the optimal solution places the most probability mass on intervals that contain the set $\{4,5\}$, as coordinating on these time-steps eliminates the solution shown in Figure~\ref{fig:uncoord_1_T_6}.

When we extend the horizon so that $T = 10$, we observe an even more diverse set of uncoordinated trajectories with high cost, which necessitate paying for coordination over longer time intervals.
Indeed, in Figure~\ref{fig:example_T_10}, we depict a subset of the first-order solutions that we find, and in the optimal solution to the coordination problem, we observe that the team pays to coordinate on the interval $\{3,4,5,6,7,8\}$.
That is, in longer-horizon settings, uncoordinated behavior is more diverse and thus coordination for longer time horizons is more valuable.

\begin{figure}
    \centering
    \includegraphics[width=0.5\linewidth]{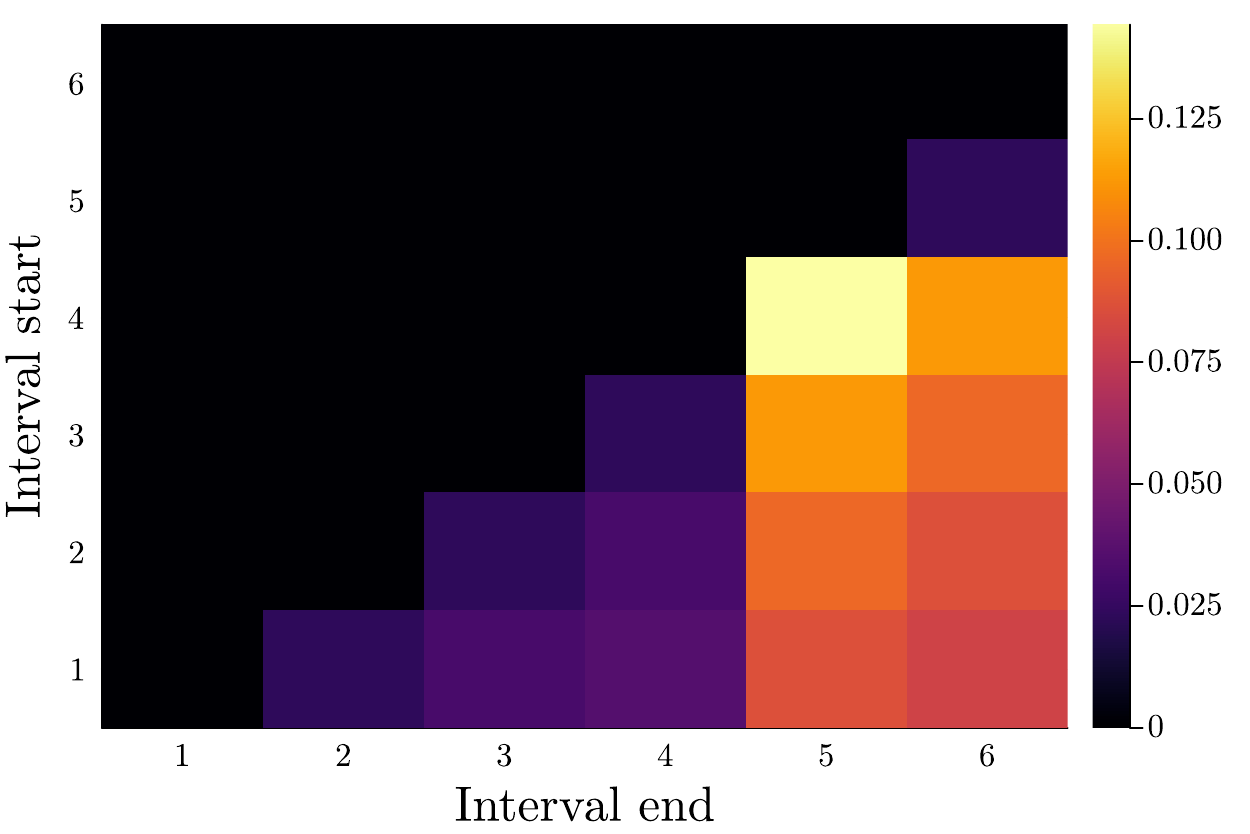}
    \caption{\textbf{The optimal coordinated solution finds the smallest interval that differentiates coordinated and uncoordinated behavior.}
    We depict the optimal distribution $p^*$ for the problem in \eqref{eq:dyn_stack_coordination} when $T=6$, where we color the heatmap by the value of $p^*$ for each interval.
    We observe that the optimal solution places probability mass on intervals that contain $\{4,5\}$, with higher mass being placed on shorter intervals. 
    This solution reflects the fact that, in order to avoid high-cost uncoordinated solutions, the team only needs to be optimal in the decision variables corresponding to  time-steps $4$ and $5$.
    }
    \label{fig:heatmap_T_6}
\end{figure}

\begin{figure}  
     \centering
     \begin{subfigure}[b]{0.48\textwidth}
         \centering
         \includegraphics[width=\textwidth]{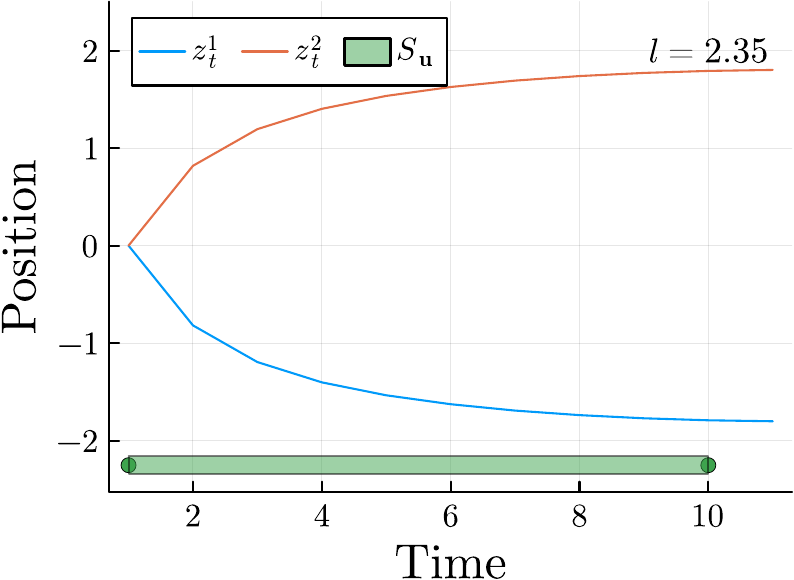}
         \caption{}
         \label{fig:coord_t_10}
     \end{subfigure}
     \hfill
     \begin{subfigure}[b]{0.48\textwidth}
         \centering
         \includegraphics[width=\textwidth]{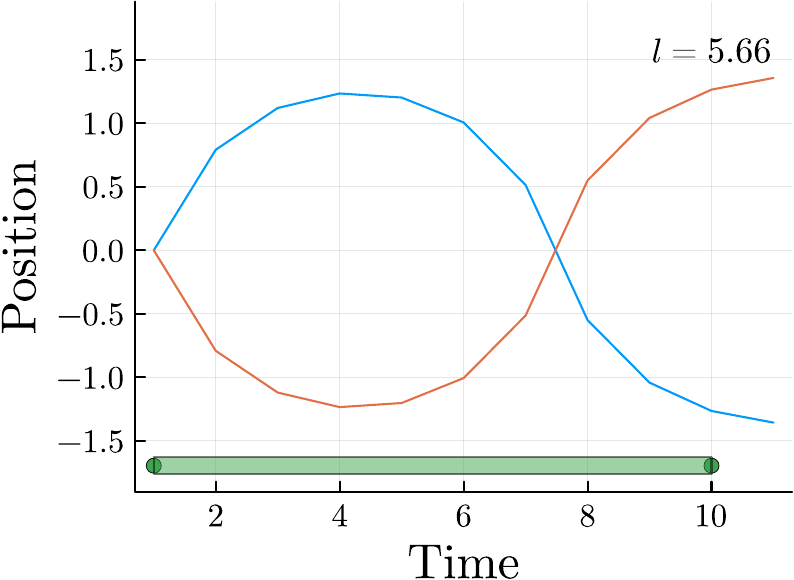}
         \caption{}
         \label{fig:coord_t_10_2}
     \end{subfigure}
     \begin{subfigure}[b]{0.48\textwidth}
         \centering
         \includegraphics[width=\textwidth]{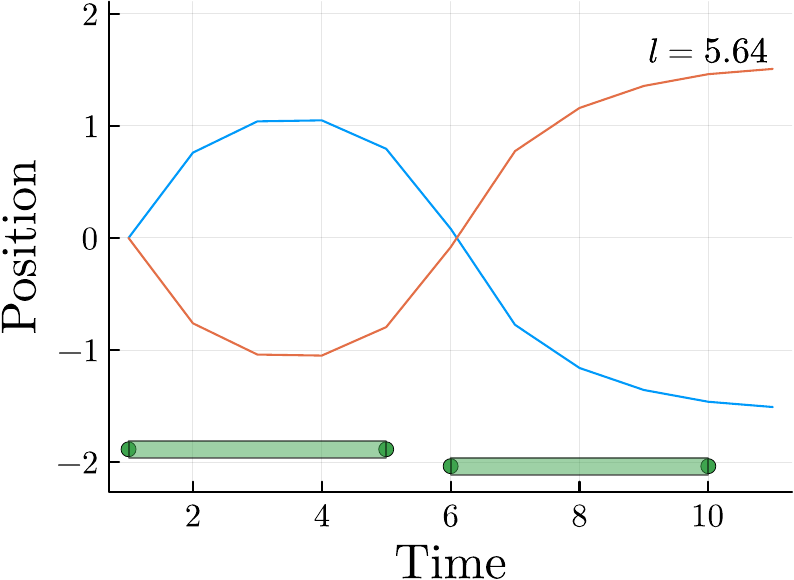}
         \caption{}
         \label{fig:uncoord_t_10_1}
     \end{subfigure}
     \hfill
     \begin{subfigure}[b]{0.48\textwidth}
         \centering
         \includegraphics[width=\textwidth]{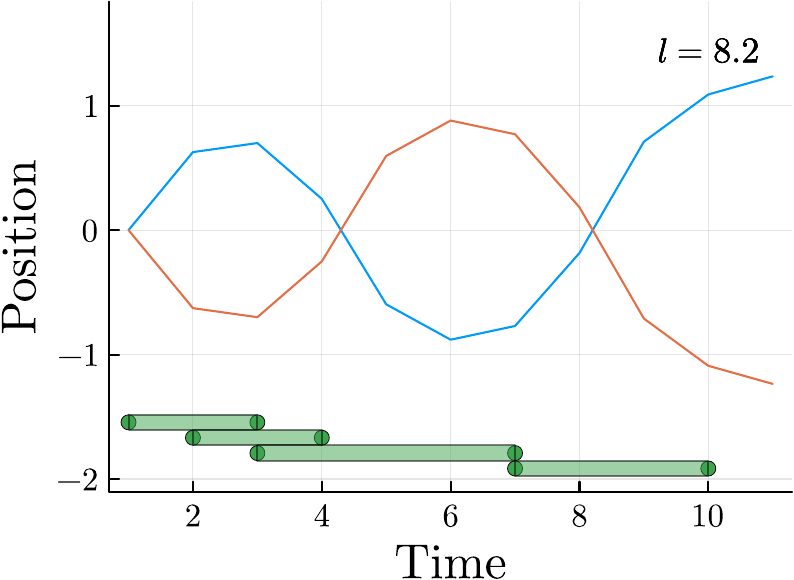}
         \caption{}
         \label{fig:uncoord_t_10_2}
     \end{subfigure}
     \par\bigskip
     \begin{subfigure}[b]{0.6\textwidth}
         \centering
         \includegraphics[width=\textwidth]{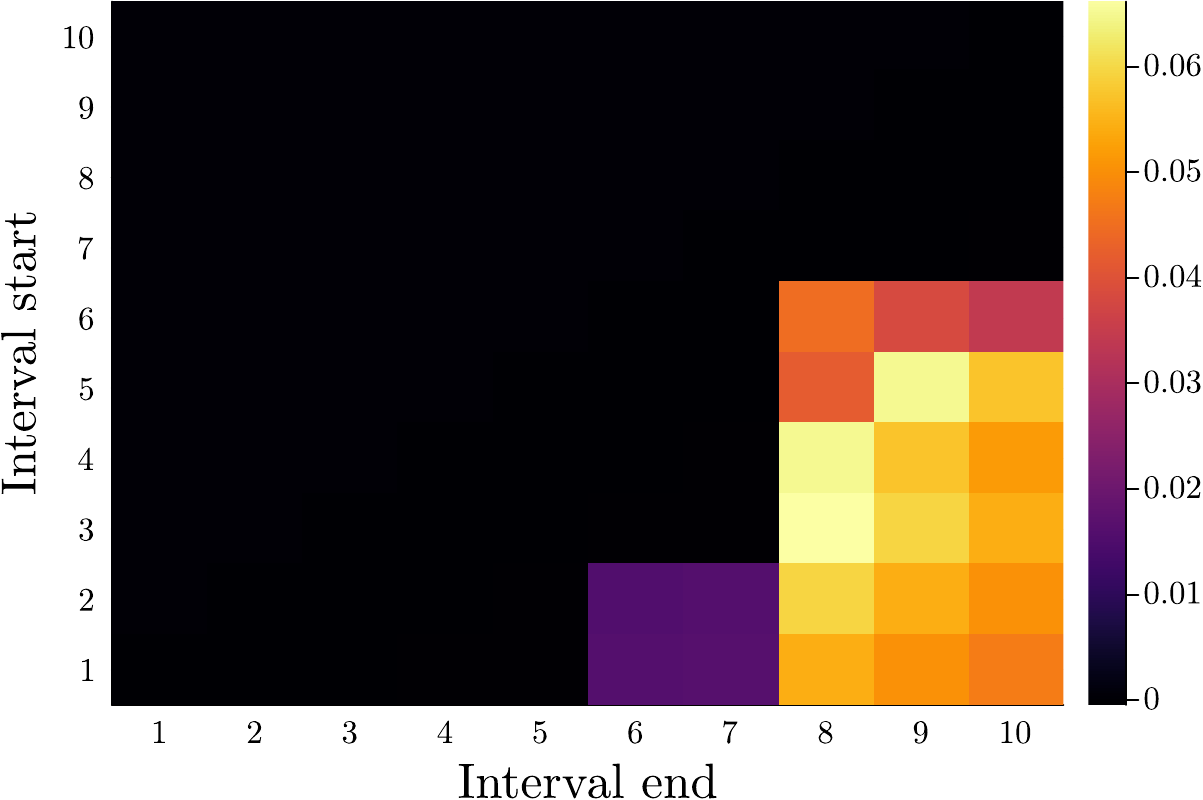}
         \caption{}
         \label{fig:t_10_heatmap}
     \end{subfigure}
        \caption{\textbf{Coordination is necessary in longer horizon settings to handle diverse uncoordinated solutions.}
        In (a-d), we plot $z_t$ for four of the first-order solutions we generate in the case where the horizon $T$ is $10$.
        We additionally depict the optimal solution to the coordination problem \eqref{eq:dyn_stack_coordination} in (e).
        We observe a diverse set of behaviors among the first-order solutions, and we observe that, in the solution to the coordination problem \eqref{eq:dyn_stack_coordination} for $T = 10$, the team pays to place high probability mass on the interval $\{3,4,5,6,7,8\}$. 
        We remark that the behavior where agents swap positions twice, as in (d), does not appear among the coordinated solutions we find.}
        \label{fig:example_T_10}
\end{figure}

\section{Conclusion}

Coordination is critical in multi-agent systems.
Uncoordinated team decision-making can lead to catastrophic behaviors, and yet constant communication is costly and, as we demonstrate in several examples, unnecessary.
To this end, we have presented a framework for studying coordination in smooth motion-planning problems.
In particular, we model a spectrum of coordinated behavior with joint team decision making at one end, and decentralized, individually-rational decision making encoded as Nash equilibria at the other.
We have shown that we can differentiate solutions with different levels of coordination by examining the second-order properties of agents' objectives, and we have used this reasoning to design an algorithm for deciding on the times at which coordination is most valuable in differentiable motion-planning scenarios.

\section{Omitted Proofs}

\textit{Proof of Lemma~\ref{lem:uncoord_connec}}
For clarity, we change notation such that $z(t) = c(t)$.
By assumption, for each $t \in [0,1]$, we have
\begin{equation}
    \nabla_z f = - \sum_i \lambda_i(t) \nabla_z g_i(z(t)) - \sum_j \mu_j(t) \nabla_z h_j(z(t)).
\end{equation}
We will argue that $\langle \dot{z},\nabla_z f(z(t)) \rangle$ is zero for $t \in (0,1)$.

Fix some $j \in \{1,\ldots,\eqdim\}$ as an index for an equality constraint. 
First, note that, as $z(t)$ is feasible for all $t$, we have $\nicefrac{dh_j(z(t))}{dt} = 0$, and thus
\begin{equation}
    \langle \dot{z},\nabla_z h_j(z(t)) \rangle = \frac{dh_j(z(t))}{dt} = 0.
\end{equation}
We next claim that $\lambda_i(t) \langle \dot{z} , \nabla_z g_i(z(t)) \rangle = 0$ for all $t \in (0,1)$ and $i \in \{1,\ldots,\ineqdim\}$.
Indeed, if for some $t_0$ we have $g_i(z(t_0)) > 0$, then by complementarity, we have $\lambda_i(t_0) = 0$.
If instead, we have $g_i(z(t_0)) = 0$, then we can deduce that $t_0$ is a local minimum of the function $t \mapsto g_i(z(t))$ and thus we have
\begin{equation}
    \langle \dot{z} , \nabla_z g_i(z(t)) \rangle = \frac{dg_i(z(t))}{dt} =  0.
\end{equation}
We can now evaluate $ \langle\dot{z} , \nabla_z f(z(t))\rangle$ to obtain
\begin{align}
    \langle \dot{z},\nabla_z f(z(t)) \rangle & = - \sum_i \lambda_i(t) \langle \dot{z}, \nabla_z g_i(z(t)) \rangle - \sum_j \mu_j(t) \langle \dot{z},\nabla_z h_j(z(t)) \rangle \\
    & = 0,
\end{align}
for all $t \in (0,1)$.
\hfill $\square$

\begin{credits}
\subsubsection{AI tool usage.} AI tools were used (via Grammarly and ChatGPT) for editing grammar and improving the flow of writing throughout the paper. 

\subsubsection{\ackname}
Caleb Probine, Su Ann Low, and Ufuk Topcu were supported by 
the Office of Naval Research, under grant number N00014-25-1-2479, 
the Army Research Office, under grant number W911NF-23-1-0317,
and the National Science Foundation, under grant number 2211432.
David Fridovich-Keil was supported by the Army Research Laboratory (cooperative agreement number W911NF-25-2-0021) and by the National Science Foundation (grant numbers 2211548 and 2336840).

\subsubsection{\discintname}
The authors have no competing interests to declare that are relevant to the content of this article.
\end{credits}

\bibliographystyle{splncs04}
\bibliography{refs}

\end{document}